\documentclass[conference]{IEEEtran}
\usepackage{cite}
\usepackage{amsmath,amssymb,amsfonts}
\usepackage{amsthm}
\newtheorem{theorem}{Theorem}

\newtheorem{proposition}{Proposition}[theorem]
\usepackage{bm}
\usepackage{algorithm}
\usepackage{algpseudocode}
\algdef{SE}[SUBALG]{Indent}{EndIndent}{}{\algorithmicend\ }%
\algtext*{Indent}
\algtext*{EndIndent}
\usepackage{graphicx}
\usepackage{textcomp}
\usepackage{xcolor}

\mathchardef\mhyphen="2D

\newcommand\Vector[1]{\bm{#1}}

\newcommand\vb{{\Vector{b}}}

\newcommand\vg{{\Vector{g}}}
\newcommand\vh{{\Vector{h}}}

\newcommand\vp{{\Vector{p}}}
\newcommand\vq{{\Vector{q}}}

\newcommand\vs{{\Vector{s}}}

\newcommand\vv{{\Vector{v}}}
\newcommand\vw{{\Vector{w}}}
\newcommand\vx{{\Vector{x}}}
\newcommand\vy{{\Vector{y}}}

\newcommand\vtheta{{\Vector{\theta}}}

\newcommand\vrho{{\Vector{\rho}}}

\newcommand\sH{{\mathbb{H}}}
\newcommand\sI{{\mathbb{I}}}

\newcommand\sU{{\mathbb{U}}}

\newcommand\Prob{\mathbf{Prob}}
\newcommand\givenp[1][]{\:#1\vert\:}
\newcommand\SNR{\mathrm{SNR}}
\newcommand\nSNR{\overline{\mathrm{SNR}}}

\newcommand\E{\mathbb{E}}

\newcommand\R{\mathbb{R}}

\newcommand\Var{\mathrm{Var}}

\DeclareMathOperator*{\argmin}{arg\,min}

\DeclareMathOperator{\sign}{sign}

\begin{document}

\title{Cluster-Based Cooperative Digital Over-the-Air Aggregation for
Wireless Federated Edge Learning
}

\author{
\IEEEauthorblockN{Ruichen Jiang, Sheng Zhou}
\IEEEauthorblockA{Beijing National Research Center for Information Science and Technology}
\IEEEauthorblockA{Department
of Electronic Engineering, Tsinghua University,
Beijing 100084, P.R.~China}
\IEEEauthorblockA{Emails: jrc16@mails.tsinghua.edu.cn, sheng.zhou@tsinghua.edu.cn}
}

\maketitle

\begin{abstract}
In this paper, we study a federated learning system at the wireless edge
that
uses over-the-air computation (AirComp).
In such a system, users transmit their messages over a multi-access channel
concurrently to achieve fast model aggregation.
Recently, an AirComp scheme based on digital modulation
has been proposed featuring one-bit
gradient quantization and truncated channel inversion at users
and a majority-voting based decoder at the
fusion center (FC).
We propose an improved digital AirComp scheme to relax its requirements on the transmitters,
where users perform phase correction and transmit with full power.
To characterize the
decoding failure probability at the FC,
we introduce
the normalized detection signal-to-noise ratio (SNR), which can be interpreted as the
effective participation rate
of users.
To mitigate wireless fading, we further propose a cluster-based
system and design the relay selection scheme based on the normalized
detection SNR. By local data fusion within each cluster and relay selection,
our scheme can fully exploit spatial diversity to increase the effective
number of voting users and accelerate model convergence.
\end{abstract}


\section{Introduction}
With the advent of Internet of Things (IoT),
increasing number of mobile devices---such as
smart phones, wearable devices and wireless sensors---will contribute data to wireless systems.
The massive data distributed over those devices
provide opportunities as well as challenges for
deploying data-driven machine learning applications.
Recently, researchers propose the concept of edge learning to
achieve fast access to the enormous data on edge devices,
where the deployments of machine learning models are pushed from the cloud to
the network edge \cite{intelligent_edge}.
Federated edge learning \cite{federated_l} is a popular framework for distributed machine learning
tasks. Under the coordination of a fusion center (FC), it can fully utilize the massive datasets while protecting
the user's privacy. In such a framework, the raw data is kept locally and users only upload
model updates to the FC, who then broadcasts the aggregated model to all users.
However, compared with the limit bandwidth available at the network edge,
current machine learning models consist of a huge number of parameters, which
makes the communication overhead the main bottleneck for federated edge learning \cite{federated_l}.

To tackle this problem, over-the-air aggregation (also
called over-the-air computation, or AirComp) is introduced in
\cite{broadband_aircomp, fl_fading,obda}. The basic idea
is to utilize the superposition property of the wireless multi-access channel (MAC)
to aggregate the uploaded model updates over the air. Compared with the traditional
system where the MAC is partitioned into orthogonal channels to ensure no interference,
such approach can achieve much shorter latency and hence faster model
convergence \cite{broadband_aircomp, fl_fading,obda}. However,
most AirComp systems in the literature use uncoded analog modulation, which can be
difficult to implement in today's widely-used digital communication systems.
Recently, an AirComp scheme using digital modulation is proposed in \cite{obda}
by integrating AirComp with a specific learning algorithm called signSGD with majority vote \cite{signsgd}.
However, a common drawback of all the schemes mentioned above is that they require
full channel side information (CSI) and the ability to arbitrarily adjust signal power
at the transmitters, which can be impractical for large-scale networks with simple nodes.

Our key observation in this paper is that while such requirements are necessary for analog AirComp in order to achieve magnitude alignment at the receiver,
this is not the case for the digital AirComp
proposed in \cite{obda}.
Specifically, we propose an improved digital AirComp scheme based on phase correction,
where users first correct the phase shift
caused by wireless fading before sending the signals over the MAC with full power.
The received signals at the FC are the sums of users' messages weighted by
channel gains,
which is decoded in a majority-vote fashion.
Compared with previous works, our scheme only requires
the phases of channel gains at the transmitters and does not need
adaptive power control.
To analyze the decoding failure probability, the
normalized detection signal-to-noise ratio (SNR) is introduced.
We find that the degradation caused by wireless fading can be regarded as reducing
the number of participating users.
A cluster-based system is proposed
to mitigate such effect by exploiting spatial diversity.
By local data fusion and relay selection,
our system can reduce the failure probability to a level close to the ideal system with perfect channels.
The
simulations validate that our scheme can increase the effective participation
rate of users and accelerate model convergence.

\emph{Notations}: We use $\vx[i]$ to denote the $i$-th component of a vector $\vx$,
$\mathcal{CN}(0,\sigma^2)$ to denote a zero-mean circularly symmetric complex Gaussian
distribution with variance $\sigma^2$, $\mathcal{N}(0,\sigma^2)$ to denote a
zero-mean real Gaussian distribution with variance $\sigma^2$, and $\mathrm{j}\mkern1mu$ to denote the imaginary unit.

\section{System Models}
\subsection{Learning Model}
The learning task we consider is to train a
common machine learning model on datasets
distributed among $K$ users in the network.
It can be formulated as
\begin{equation}\label{eq:erm model}
    \min_{\vw\in \R^d}F(\vw) := \frac{1}{K}\sum_{k=1}^K f_k(\vw),
\end{equation}
where $f_k(\vw)$ denotes the local loss function
parameterized by $\vw\in \R^d$ with respect to the dataset on user $k$.
We will refer to $F(\vw)$ as the global loss function.

To solve \eqref{eq:erm model}, a straightforward approach is letting
all users upload their datasets and running centralized learning algorithms
at the FC. However, not only will this raise privacy concerns for
datasets involving sensitive information,
but also will consume excessive bandwidth and hence
incur unbearable delays.
Therefore, we employ the federated learning
framework \cite{federated_l} 
using signSGD with majority vote as shown in
Algorithm \ref{alg:signsgd}, which is
a communication-efficient
distributed learning algorithm proposed in \cite{signsgd}.
Specifically, at iteration round~$n$, user~$k$
computes the stochastic gradient $\vg_k(\vw^{(n)})\in \R^d$ at
the current model parameters $\vw^{(n)}$ and sends their signs to the FC. The FC then aggregates the data by majority vote:
we may regard each user's local update as ``voting" between $\{+1,\,-1\}$ for each gradient component,
and the global update $\tilde{\vg}^{(n)}\in \R^d$ is determined by what most users agree with.
Finally, the FC broadcasts $\tilde{\vg}^{(n)}$ to all users who then
update the parameter vector as in line \ref{line:update}.

With proper choices of learning rate and
local batch size, it is shown in \cite{signsgd} that Algorithm~\ref{alg:signsgd} can
achieve a similar convergence rate as distributed SGD while greatly reducing the
communication cost. However, their analysis assumes
ideal channels between
the FC and all users, which is unrealistic
in practical wireless communication systems. In the next section,
we will model the MAC between the FC and users to
take wireless impairments into account.

\begin{algorithm}[t]
    {
    \caption{signSGD with majority vote\cite{signsgd}}\label{alg:signsgd}
    \begin{algorithmic}[1]
      \State \textbf{Input:} learning rate $\eta$, initial model parameters $\vw^{(0)}$
      \For{iteration round $n=1,2,\ldots,N$}
      \State \textbf{Users:}
      \Indent
      \For{$k=1,2,\ldots,K$ in parallel}
      \State Compute the stochastic gradient $\vg_k(\vw^{(n)})$
      \State Send $\sign(\vg_k(\vw^{(n)}))$ 
      \EndFor
      \EndIndent
      \State \textbf{FC:}
      \Indent
      \State Computes $\tilde{\vg}^{(n)} \leftarrow \sign\big[\sum_{k=1}^K \sign(\vg_k(\vw^{(n)}))\big]$ 
      \State Broadcasts $\tilde{\vg}^{(n)}$ to all users
      \EndIndent
      \State \textbf{Users:}
      \State Update model parameters by $\vw^{(n+1)}\leftarrow\vw^{(n)}-\eta \tilde{\vg}^{(n)}$  \label{line:update}
      \EndFor
    \end{algorithmic}
    }
\end{algorithm}

\subsection{Communication Model}\label{subsec:communication_model}
We consider a single-cell network where locations of users are
uniformly distributed in the circle centered at the FC
with radius $R$.
Therefore, the distances between users
and the FC are independent and identically distributed (i.i.d.)
random variables, whose probability distribution function (pdf) is given by
\begin{equation}\label{eq:location_dist}
    f(r) = \frac{2r}{R^2},\;0<r<R.
\end{equation}
To capture the effect of large-scale fading,
we assume the path loss at distance $r$ is
    \begin{equation}\label{eq:path_loss}
        P_L(r) =
        \begin{cases}
            1, & r\leq r_0;\\
            (\frac{r}{r_0})^{-\alpha}, & r\geq r_0.
        \end{cases}
    \end{equation}
Here, $P_L(r)$ is set to a constant within $r_0$
to avoid singularity, and
$\alpha\in (2,4)$ is the path loss exponent.
Moreover, the MAC also suffers from frequency-selective
fading and we employ orthogonal frequency division multiple access (OFDMA)
to divide the broadband channel into $M$ subchannels.

For ease of exposition, we use binary phase shift keying (BPSK) on each subchannel,
and it can be readily extended to 4-ary quadrature amplitude modulation (4-QAM).
Note that in Algorithm \ref{alg:signsgd}, the FC
is interested in the outcome of majority vote $\tilde{\vg}^{(n)}$ rather than the individual
local update at each user. This motivates us to use uncoded transmission, where
the gradient signs are mapped directly into BPSK symbols, and let all users
send their signals over the MAC concurrently.
Denote the transmitted symbols from user $k$ at iteration round $n$ by
$\vs^{(n)}_k := \sign(\vg_k(\vw^{(n)}))$, then the received signals $\vy^{(n)}[i]$ ($i=1,2,\ldots,d$) at the FC are given by
\begin{equation}\label{eq:aircomp}
        \vy^{(n)}[i]=\sum_{k=1}^K \sqrt{P_L(r^{(n)}_k)} \vh^{(n)}_{k}[i] \vb^{(n)}_k[i] \vs^{(n)}_k[i]+\vv^{(n)}[i],
\end{equation}
where $r^{(n)}_k$ is the distance between user $k$ and the FC,
$\vh^{(n)}_{k}$ is the wireless fading coefficient,
$\vb^{(n)}_k$ is the amplifying coefficient at user $k$ and $\vv^{(n)}$ is the additive
Gaussian white noise. Statistically, we assume that $\{\vh^{(n)}_{k}[i]\}$ are
independent and identically distributed (i.i.d.) according to $\mathcal{CN}(0,1)$ and
$\{\vv^{(n)}[i]\}$ i.i.d. according to $\mathcal{CN}(0,N_0)$.
Moreover, we impose a power budget $P$ on each OFDM symbol of each user.
Since one OFDM symbol consists of $M$ BPSK symbols, this implies
\begin{equation*}
    |\vb^{(n)}_k[i]|^2\leq P_s :=\frac{P}{M},\;i=1,2,\ldots,d.
\end{equation*}
Finally, the FC uses the detected BPSK symbols as the aggregated results:
\begin{equation}\label{eq:bpsk_detect}
    \tilde{\vg}^{(n)}[i] = \sign (\vy[i]),\;i=1,2,\ldots,d.
\end{equation}

Such idea of integrating AirComp with signSGD with majority vote is first proposed in \cite{obda},
where truncated channel inversion is adopted to achieve magnitude alignment at the FC.
However, note that in \eqref{eq:bpsk_detect} the FC only preserves the
signs of the received signals, suggesting that
the signals from different users should be aligned in phase
but not necessarily in amplitude. Hence, we let all users
compensate the phase drift caused by $\vh^{(n)}_{k}[i]$ and transmit
their signals in full power, which means
\begin{equation}\label{eq:phase_correction}
    \vb^{(n)}_k[i] = \sqrt{P_s}\exp(-\mathrm{j}\mkern1mu \vtheta_k^{(n)}[i]),
\end{equation}
where $\vtheta_k^{(n)}[i]=\arg(\vh^{(n)}_{k}[i])$. We refer to our scheme as
digital AirComp based on phase correction (digital AirComp-PC). From \eqref{eq:phase_correction},
we can see that such scheme does not
require full CSI or arbitrary power adjustment at the transmitters.


It is worth noting that
digital AirComp-PC is mathematically equivalent to the equal gain combining
(EGC) scheme\cite{decision_fusion} in the distributed detection literature.
In \cite{decision_fusion}, the authors consider aggregating decisions
from wireless sensor nodes over Rayleigh fading channels and the EGC scheme
is shown to be robust
for most SNR range. The difference is that there the phase correction and
signal aggregation are performed at the FC, while in ours
the phase correction is performed at the transmitters and the signal aggregation
is achieved over the air by the superposition property of wireless channels.
This enables us to attain higher spectral efficiency with similar performance guarantees.

\subsection{Cluster-based Cooperative Model}\label{subsec:cluster}
Our system model is shown in Fig.~\ref{fig:relay_cooperation}.
First of all, at the beginning of iteration round $n$, we assume the users are already
divided into $C$ clusters denoted by $\sU_1^{(n)},\sU_2^{(n)},\ldots,\sU_C^{(n)}$,
which satisfies $\cup_{c=1}^C \sU^{(n)}_c = \{1,2,\ldots, K\}$
and $\sU^{(n)}_i\cap\sU^{(n)}_j=\emptyset,\;\forall i\neq j$.
In each cluster, $L$ users are chosen as relays uniformly and randomly.
The majority vote consists of the following two steps:
\begin{enumerate}
    \item Within cluster $\sU^{(n)}_c$, we
    select one of the relays as the local fusion center for each component of the
    gradients. The relay selection scheme will be discussed in Section \ref{sec:relay_cooperation}.
    Let
    $\sI^{(n)}_{c,l}$ denote the indices of components collected by relay $l$, then we have
    $\cup_{l=1}^L \sI^{(n)}_{c,l} = \{1,2,\ldots, d\}$ and
    $\sI^{(n)}_{c,i}\cap\sI^{(n)}_{c,j}=\emptyset,\;\forall i\neq j$.
    For every component $i\in \sI^{(n)}_{c,l}$, relay $l$ collects the result of
    majority vote in cluster $\sU^{(n)}_c$ as
    \begin{equation}\label{eq:majority_vote_within_cluster}
        \tilde{\vs}^{(n)}_c[i] = \sign\big(\sum_{k\in \sU^{(n)}_c} \vs^{(n)}_k[i] \big),
    \end{equation}
    where $\vs^{(n)}_k[i]\in \{+1,-1\}$ is the sign of the $i$-th component of the local stochastic
    gradient at user $k$.
    \item The selected relay nodes report the fused result $\tilde{\vs}^{(n)}_c[i]$
    to the FC using the digital AirComp-PC scheme
    discussed in Section \ref{subsec:communication_model}.
\end{enumerate}
\begin{figure}[t]
    \centering
    \includegraphics[width=0.8\linewidth]{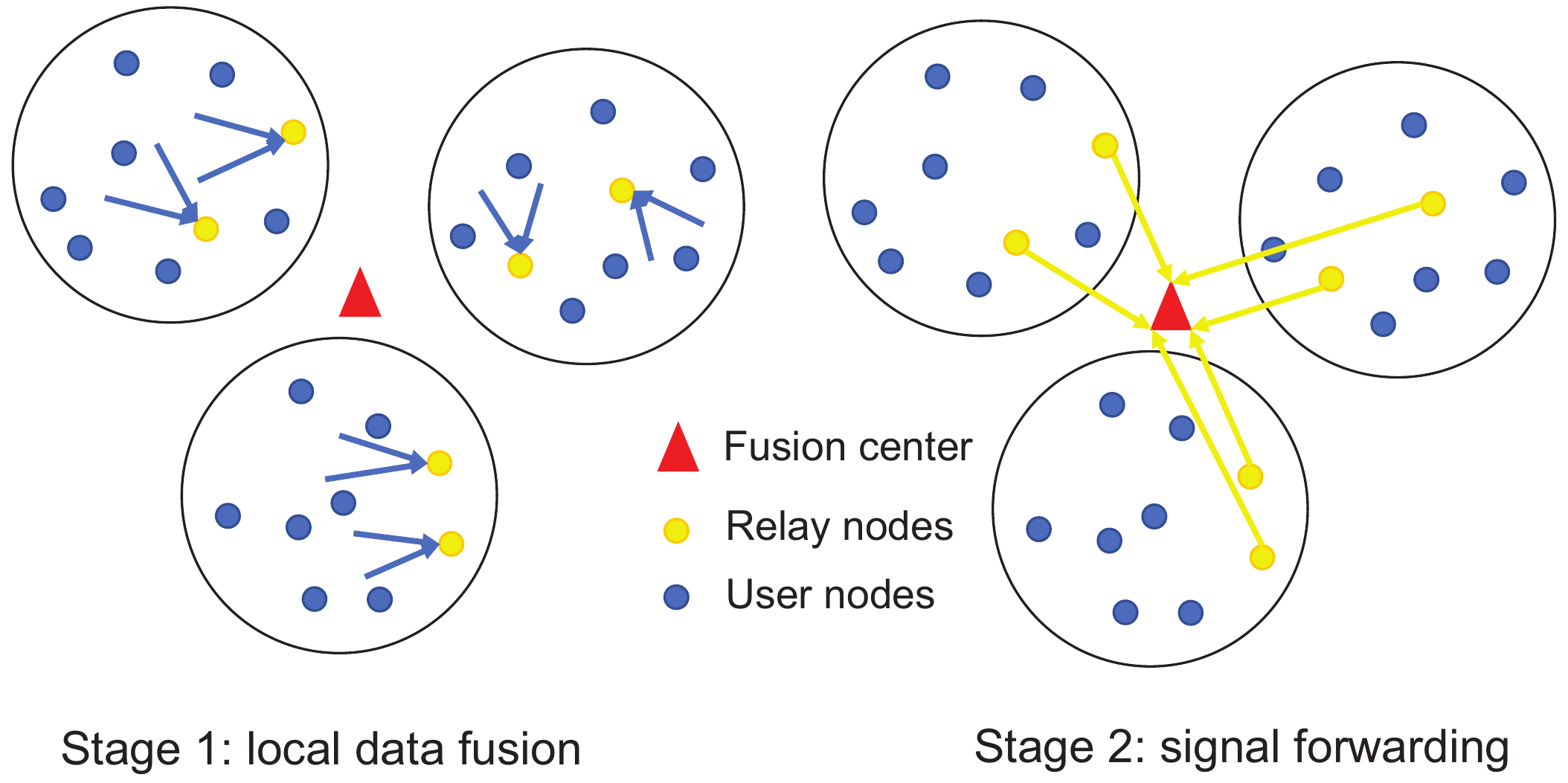}
    \caption{The illustration of cluster-based cooperation.}
    \label{fig:relay_cooperation}
\end{figure}
For simplicity of analysis, we make the following assumptions:
\begin{enumerate}
    \item All the user clusters are of equal size $K_C$, i.e., $|\sU^{(n)}_1|=|\sU^{(n)}_2|\ldots=|\sU^{(n)}_C|=K_C$.
    \item The local decision fusion within each cluster is perfect.
    \item The $L$ relays in any cluster are the same distance away from the FC, denoted by
    $r_c^{(n)}$. Moreover, $r_c^{(n)}$ are i.i.d. random variables with the pdf given in \eqref{eq:location_dist}.
    \item The FC has the instantaneous CSI of the channels between
    the relay nodes in each cluster and the FC,
    which is denoted by $\sH^{(n)}_c := \{\vrho^{(n)}_{c,1},\vrho^{(n)}_{c,2},\ldots,\vrho^{(n)}_{c,L}\}\,(c=1,2,\ldots,C)$.
    Here $\vrho^{(n)}_{c,l}:=\sqrt{P_L(r^{(n)}_c)}|\vh^{(n)}_{c,l}|\;(l=1,2,\ldots,L)$ are the channel gains and
    $\vh^{(n)}_{c,l}\sim \mathcal{CN}(0,1)$ are the Rayleigh fading coefficients between relay $l$ and the FC.
    The FC makes the relay selection based on $\sH^{(n)}_c$
    and notifies the selected relay nodes.
\end{enumerate}
Similar to \eqref{eq:aircomp}, the received signal $\vy$ can be written as
\begin{equation*}
    \vy[i]=\sum_{c=1}^C \vrho^{(n)}_c[i] \sqrt{P_s}  \tilde{\vs}^{(n)}_c[i]+\vv^{(n)}[i],\;\;\;i=1,2,\ldots,d,
\end{equation*}
where $\vrho^{(n)}_c[i]\in \sH^{(n)}_c \cup \{0\}$ depends on our relay selection scheme and
$\vv^{(n)}[i]\sim \mathcal{CN}(0,N_0)$ is the i.i.d. additive white Gaussian noise.

It is worth noting that similar cluster-based methods have been applied to cooperative spectrum sensing
to merge sensing observations from different users
\cite{cluster_cooperative_spectrum_sensing,Energy_efficient_tech}.
However, their focuses are on the false alarm probability or energy efficiency
while ours are on the failure probability that will be introduced in Section \ref{sec:fail_prob_analysis}.

\section{Failure Probability Analysis}\label{sec:fail_prob_analysis}
Prior works have shown that the convergence speed of
Algorithm \ref{alg:signsgd} crucially depends on the decoding failure probability $\vq^{(n)}[i]$ ($i=1,2,\ldots,d$)
defined as
\begin{equation}
    \vq^{(n)}[i] = \Prob\big[\tilde{\vg}^{(n)}[i]\neq \sign\big(\nabla F(\vw^{(n)})[i]\big) \givenp \vw^{(n)} \big],
\end{equation}
which is the probability that the result of majority vote mismatches the sign of the true gradient.
In particular, smaller failure probability leads to a higher convergence speed,
and increasing the number of participating users can attain similar speedup to distributed SGD \cite{signsgd}.
In this section, we will discuss the failure probability of the digital AirComp-PC scheme proposed in
Section \ref{subsec:communication_model} and show that
wireless impairments effectively decrease the participating rate of voting users.

In our scheme, the failure probability depends on the following three factors:
\begin{enumerate}
    \item The local success probability defined as
    \begin{equation*}
        \vp^{(n)}_k[i]:=\Prob\big[\vs^{(n)}_k[i]= \sign\big(\nabla F(\vw^{(n)})[i]\big) \givenp \vw^{(n)}\big],
    \end{equation*}
    which is the probability that the sign of the local stochastic gradient
    is the same as the sign of the true gradient;
    \item The number of voting users $K$;
    \item The wireless channel conditions between the FC and users.
\end{enumerate}
We first discuss the local success probability. Intuitively, at the beginning stage of
model training, the global loss function $F$ has a relatively large gradient at
the point $\vw^{(n)}$, and hence with high probability the sign of the
local stochastic gradient is correct. As the model training proceeds,
$\vw^{(n)}$ will gradually approach the stationary point of
$F$ and the true gradient will be close to zero. This means that the noise terms will dominate
the local stochastic gradients resulting in almost random guess at each user.
Therefore, we
will set the local success probability to near 0.5 and focus on the other
two factors in the following.

\subsection{Detection SNR} \label{subsec:detection SNR}
Because of symmetry, we only need to analyze the failure probability of
a specific gradient component $\nabla F(\vw^{(n)})[i]$.
Assume $\nabla F(\vw^{(n)})[i]>0$ without loss of generality.
To simplify notation, we omit the superscript
$(n)$ and the index $[i]$ in the following. The model can be now expressed as
\begin{equation}\label{eq:mac}
    y = \sum_{k=1}^K \rho_k s_k+ v,
\end{equation}
where $y$ is the received signal at the FC, $\rho_k$ is the channel coefficient
between the FC and user $k$, $s_k\in\{+1,-1\}$ is the sign of the stochastic
gradient at user $k$, and $v\sim \mathcal{N}(0, N_0/(2P_s))$ is the additive white Gaussian noise.
By our assumptions, $\{s_k\}_{k=1}^K$ are independent from each other and
$\Prob[s_k = 1] = p_k,\;\Prob[s_k = -1] = 1-p_k$ where $p_k>1/2$.
The failure probability is given by $q=\Prob{[y<0]}$.

We first assume the channel coefficients $\vrho$ are given parameters.
In general, the closed form of $q$ is intractable.
Hence, we provide an upper bound by using the properties of sub-Gaussian
random variables \cite{sub_gaussian}.

\begin{theorem}\label{thm:chernoff}
    Define
    \begin{equation}\label{eq:sigma}
        \tau^2 : = \sum_{k=1}^K \rho_k^2+\frac{N_0}{2P_s}.
    \end{equation}
    If $y$ is given by \eqref{eq:mac}, then we have
    \begin{equation}\label{eq:sub_gaussian_ineq}
        q = \Prob({y<0})\leq \exp(-\frac{\E[y]^2}{2 \tau^2}),
    \end{equation}
    where
    \begin{equation}\label{eq:y_expec}
        \E[y] = \sum_{k=1}^K \rho_k (2p_k-1).
    \end{equation}
\end{theorem}
\begin{proof}
    First recall the definition of sub-Gaussian random
    variables: a zero-mean random variable $X$ is sub-Gaussian with
    parameter $\theta$ if it satisfies
        $\E[\exp(t X)]\leq \exp(\frac{\theta^2t^2}{2})\;(\forall t\in \R)$
    and we write $X\in \mathrm{subG}(\theta^2)$.
    We will use the following two key properties \cite{sub_gaussian}:
    \begin{enumerate}
        \item \label{item:tail_bound} If ${X}\in \mathrm{subG}(\theta^2)$,
        then we have the tail bound $\Prob(X>x)\leq \exp(-\frac{x^2}{2\theta^2}),\;\forall x>0$.
        \item \label{item:weighted_sum} Suppose $X_1,X_2,\ldots,X_n$ are $n$ independent random variables
        such that $X_k\in \mathrm{subG}(\theta_k^2)$. Then the weighted sum
        satisfies $\sum_{k=1}^n a_k X_k\in \mathrm{subG}(\sum_{k=1}^n a_k^2\theta_k^2)$.
    \end{enumerate}
    We are now ready to prove \eqref{eq:sub_gaussian_ineq}. Since $v\sim \mathcal{N}(0,N_0/(2P_s))$,
    it is easy to verify $v\in \mathrm{subG}(N_0/(2P_s))$ by definition.
    Moreover, we can prove $s_i-\mathbb{E}[{s_i}]\in \mathrm{subG}(1)$ with Theorem~3.1 in \cite{sub_gaussian}.
    Therefore, we can use property \ref{item:weighted_sum}) to get $y-\mathbb{E}[y]\in \mathrm{subG}(\tau^2)$,
    where $\tau^2$ is given in \eqref{eq:sigma}. Finally, property \ref{item:tail_bound})
    leads to
    \begin{equation*}
        q = \Prob(\mathbb{E}[y]-y>\mathbb{E}[y]) \leq \mathrm{exp}(-\frac{\E[y]^2}{2\tau^2}).
    \end{equation*}

\end{proof}
In BPSK detection, we may view $\E[y]^2$ and $\Var[y]$ as
the signal and noise power respectively.
Note that $\Var{[y]}=\sum_{k=1}^K 4\rho_k^2p_k(1-p_k)+N_0/(2P_s)\leq \tau^2$ and
the equality holds when $p_1=p_2=\ldots=p_K=1/2$.
Hence, we refer to $\SNR_{\mathrm{d}}:=\E[y]^2/\tau^2$ as the \emph{detection SNR}.

In the following, we assume that users have i.i.d. datasets and $p_1=p_2=\ldots=p_K=p_{\mathrm{loc}}$.
The detection SNR becomes
\begin{equation}
        \mathrm{SNR}_{\mathrm{d}}=
        (2p_{\mathrm{loc}}-1)^2 \frac{\|\vrho\|_1^2}{\|\vrho\|_2^2+N_0/(2P_s)}.
\end{equation}
Note that for ideal noiseless channels where $N_0=0$ and $\rho_1=\rho_2=\ldots=\rho_K=1$, we have $\SNR_{\mathrm{d}}=(2p_{\mathrm{loc}}-1)^2K$.
This motivates us to define the normalized detection SNR as
\begin{equation}\label{eq:normalized_SNR}
    \overline{\mathrm{SNR}}_{\mathrm{d}} = \frac{1}{K}\frac{\|\vrho\|_1^2}{\|\vrho\|_2^2+N_0/(2P_s)}.
\end{equation}
The advantage of dealing with $\nSNR_{\mathrm{d}}$ is that it depends only on the channels between the FC and users but not on $p_{\mathrm{loc}}$.
Furthermore, we have the following result:
\begin{proposition}
    If $\nSNR_{\mathrm{d}}$ is given by \eqref{eq:normalized_SNR}, then
    \begin{equation*}
        0\leq \nSNR_{\mathrm{d}} \leq 1,
    \end{equation*}
where the second equality holds if and only if $N_0=0$ and $\rho_1=\rho_2=\ldots=\rho_K$.
\end{proposition}
\begin{proof}
    We omit the proof due to space limitations.
\end{proof}
Since $\SNR_{\mathrm{d}}$ is proportional
to the number of users $K$ for ideal noiseless channels, we can interpret $\nSNR_{\mathrm{d}}$ as the
effective participation rate of users.

\subsection{Simulation Validations}\label{subsec:failure prob with large scale}
In our model, we have $\rho_k=\sqrt{P_L(r_k)}|h_k|$ where $h_k\sim \mathcal{CN}(0,1)$ i.i.d.,
the pdf of $r_k$ is given by \eqref{eq:location_dist}, and $P_L(\cdot)$ is given by \eqref{eq:path_loss}.
When $K$ is sufficiently large, we can approximate $\nSNR_{\mathrm{d}}$ by the law of large numbers:
\begin{align}
        \overline{\mathrm{SNR}}_{\mathrm{d}} &\approx\frac{(\E[ |h_k|] \E[\sqrt{P_L(r_k)}])^2}{\E [|h_k|^2]\E[P_L(r_k)]+N_0/(2KP_s)} \nonumber
        \\
    &=
     \frac{\pi/4\; \E[\sqrt{P_L(r_k)}]^2}{\E[P_L(r_k)]+N_0/(2KP_s)},\label{eq:snr_d_pl}
\end{align}
where
\begin{equation*}
    \begin{aligned}
    \E[\sqrt{P_L(r_k)}] &= -\frac{\alpha}{4-\alpha}\left(\frac{r_0}{R}\right)^2+
    \frac{4}{4-\alpha}\left(\frac{r_0}{R}\right)^{\frac{\alpha}{2}},\\
    \E[{P_L(r_k)}] &= -\frac{\alpha}{2-\alpha}\left(\frac{r_0}{R}\right)^2+
    \frac{2}{2-\alpha}\left(\frac{r_0}{R}\right)^{{\alpha}}.
    \end{aligned}
\end{equation*}

From \eqref{eq:snr_d_pl} we can see that
when $2KP_s\E[{P_L(r_k)}]/N_0\gg 1$ the channel noise is negligible.
Therefore, when the symbol power $P_s$ exceeds a certain threshold,
$\nSNR_{\mathrm{d}}$ will reach the maximum and remain unchanged. We can observe such
phenomenon from the left figure in Fig.~\ref{fig: path loss}: with
increasing $P_s$, the failure probability will decrease until reaching the error floor.
This may be explained by the fact that the users' signals are coherently
combined at the FC, resulting in a power gain similar to transmit beamforming.
From the above, we conclude that channel noise is not the bottleneck for digital
AirComp-PC at a reasonably large SNR.

Moreover, by Theorem~\ref{thm:chernoff} we have
$q\leq \exp(-K \nSNR_{\mathrm{d}})$, meaning that the failure probability
decays at an exponential rate. However, there exists a huge
gap between the digital AirComp-PC and the ideal majority vote,
which is more evident when $K$ increases. For instance,
when $\alpha=3$ and $R/r_0=30$, from \eqref{eq:snr_d_pl} we get
${\nSNR{}}_{\mathrm{d}}\approx 10.6\%$. This shows that equivalently about only 10.6\% of users
participate.
In the right
figure in Fig.~\ref{fig: path loss},
We also plot the failure probability of the equivalent ideal majority vote (the dash line),
which is indeed close to the actual performance of the digital AirComp-PC.
\begin{figure}[!t]
    \centering
    \includegraphics[width=\linewidth]{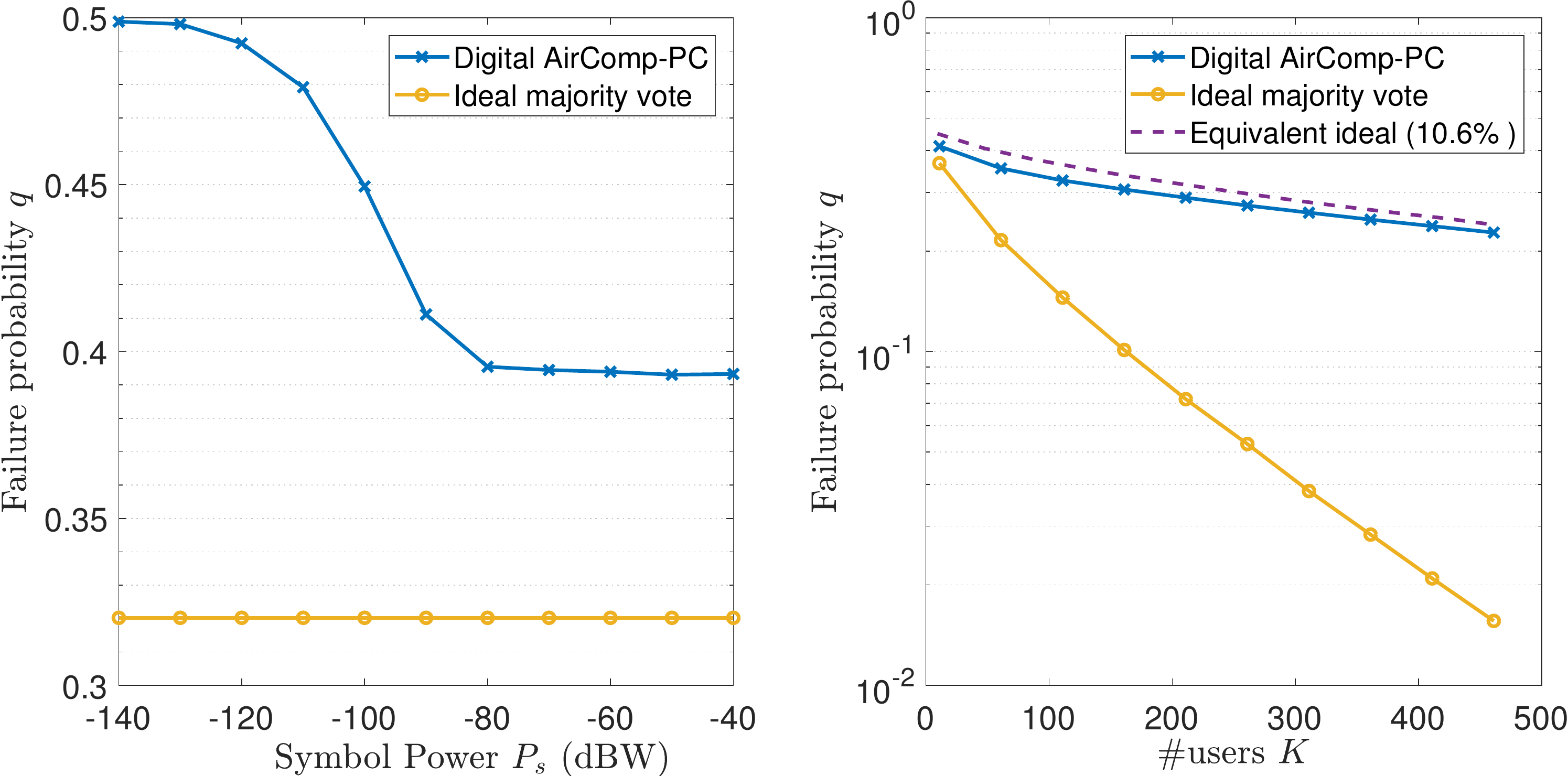}
    \caption{The left figure shows the failure probability $q$ versus symbol power $P_s$
    where $K=21,\,p_{\mathrm{loc}}=0.55$, $\alpha=3$ and $R/r_0=30$ in \eqref{eq:path_loss}.
    The right figure shows
    $q$ versus the number of users $K$ where $P_s=-50$dBW and other parameters are
    the same as the left.} 
    \label{fig: path loss}
\end{figure}

\section{Relay Selection Scheme}\label{sec:relay_cooperation}
To design the relay selection scheme in Section \ref{subsec:cluster}, we
use the normalized detection SNR as the optimization objective to minimize
the failure probability. As in Section \ref{sec:fail_prob_analysis}, we
omit the superscript and the index. The
channel model becomes
\begin{equation*}
    y = \sum_{c=1}^C \rho_c \tilde{s}_c +v,
\end{equation*}
where $\rho_c\in \sH_c\cup \{0\}$ depends on our relay selection scheme;
$\tilde{s}_c=\sign(\sum_{k\in \sU_c}s_k)$ is the majority vote result
within cluster $\sU_c$ and $\Prob[s_k=1]=p_{\mathrm{loc}},\,\Prob[s_k=-1]=1-p_{\mathrm{loc}},\,p_{\mathrm{loc}}>1/2$;
and $v\sim \mathcal{N}(0,N_0/(2P_s))$ is the additive white Gaussian noise.

When the number of users $K_c$ within each cluster is sufficiently large,
by law of large numbers we have
\begin{equation*}
    \tilde{s}_c = \frac{\sum_{k\in \sU_c} s_k}{|\sum_{k\in \sU_c}s_k|} \approx \frac{\sum_{k\in \sU_c} s_k}{K_C(2p_{\mathrm{loc}}-1)}.
\end{equation*}
Hence, the normalized detection SNR can be approximated by
\begin{align}
    {\nSNR}_{\mathrm{d}} 
    &=\frac{1}{C} \frac{\left(\sum_{c=1}^C \rho_c\right)^2}{\sum_{c=1}^C \rho_c^2+K_C(2p_{\mathrm{loc}}-1)^2N_0/(2P_s)} \nonumber\\
    &\approx \frac{1}{C}\frac{\left(\sum_{c=1}^C {\rho_c}\right)^2 }{\sum_{c=1}^C \rho_c^2}, \label{eq:snr_n_relay}
\end{align}
where we discard the noise term in \eqref{eq:snr_n_relay}. From the discussion in
Section \ref{subsec:failure prob with large scale}, we know that the channel noise
is negligible for SNR at a reasonable level.
The relay selection can be formulated as the following
optimization problem:
\begin{equation}\label{eq:opt_relay}
    \begin{aligned}
        \max_{\rho_1,\rho_2,\ldots,\rho_C}\;\; & \frac{1}{C}\frac{\left(\sum_{c=1}^C {\rho_c}\right)^2 }{\sum_{c=1}^C \rho_c^2}\\
        \mathrm{s.t.}\;\;& \rho_c \in  \{\rho_{c,1},\ldots,\rho_{c,L} \}\cup \{0\},\\
    \end{aligned}
\end{equation}
where $\rho_c=0$ corresponds to the case when no relay in cluster $\sU_c$ participates.
This is a typical discrete optimization problem with
$(L+1)^C$ feasible solutions in total. Therefore, we consider two low-complexity
relay selection schemes.

One approach is to select the relay with the strongest gain:
\begin{equation*}
    \rho_c = \max\{\rho_{c,1},\ldots,\rho_{c,L} \},\;c=1,2,\ldots,C.
\end{equation*}
We refer to it as the strongest gain scheme.
However, in the context of digital AirComp-PC, this is far from optimal. Intuitively,
because of the randomness of users' locations, the large-scale fading effects
vary greatly among the clusters. Hence, the signal sent by the relay close to the FC
will drown out the signals from other relays, resulting in the loss of effective voting users.
\begin{algorithm}[!t]
    \caption{Iterative greedy algorithm for ~\eqref{eq:opt_relay}}\label{alg:iterative_greedy}
    \begin{algorithmic}[1]
      \State \textbf{Input:} $\sH_c=\{\rho_{c,1},\ldots,\rho_{c,L}\},\;c=1,2,\ldots,C$
      \State Initialize $\rho_c \leftarrow \max \{\rho_{c,1},\ldots,\rho_{c,L}\},\;c=1,2,\ldots,C$
      \Repeat
      \For {$c=1,2,\ldots,C$}
       \State $a\leftarrow \sum_{k\neq c} \rho_k,\;b\leftarrow \sum_{k\neq c} \rho_k^2$
       \State $\rho_c \leftarrow \argmin_{\rho\in \sH_c\cup \{0\}} |\frac{1}{\rho+a}-\frac{a}{a^2+b}|$ \label{line:one_variable_update}
      \EndFor
      \Until{$\{\rho_1,\rho_2, \ldots,\rho_C\}$ are unchanged}
    \end{algorithmic}
  \end{algorithm}

To mitigate such drawback, we propose an iterative greedy algorithm as shown in Algorithm \ref{alg:iterative_greedy}.
We first use the strongest gain scheme to initialize $\rho_1,\rho_2,\ldots,\rho_C$.
Then we successively optimize one variable at a time while keeping the others fixed
and repeat the process until all the variables no longer change. We can see that
the objective value is non-decreasing after every iteration. Since there are only finite feasible solutions,
Algorithm \ref{alg:iterative_greedy} is guaranteed to terminate in finite steps.


We run numerical experiments to compare the strongest gain scheme with the iterative greedy
algorithm and plot the results in Fig.~\ref{fig:cooperation_error_curve}.
When user clustering and relay selection are adopted, we can see a substantial improvement
upon the digital AirComp-PC without cooperation. Moreover, the
iterative greedy algorithm attains a much lower failure probability than
the strongest gain scheme, and is close to the ideal case when $L=5$.
Together with the discussion in Section~\ref{subsec:failure prob with large scale},
this shows that our cluster-based cooperation scheme can effectively increase
the participation rate of voting users.
\begin{figure}[!t]
    \centering
    \includegraphics[width=0.8\linewidth]{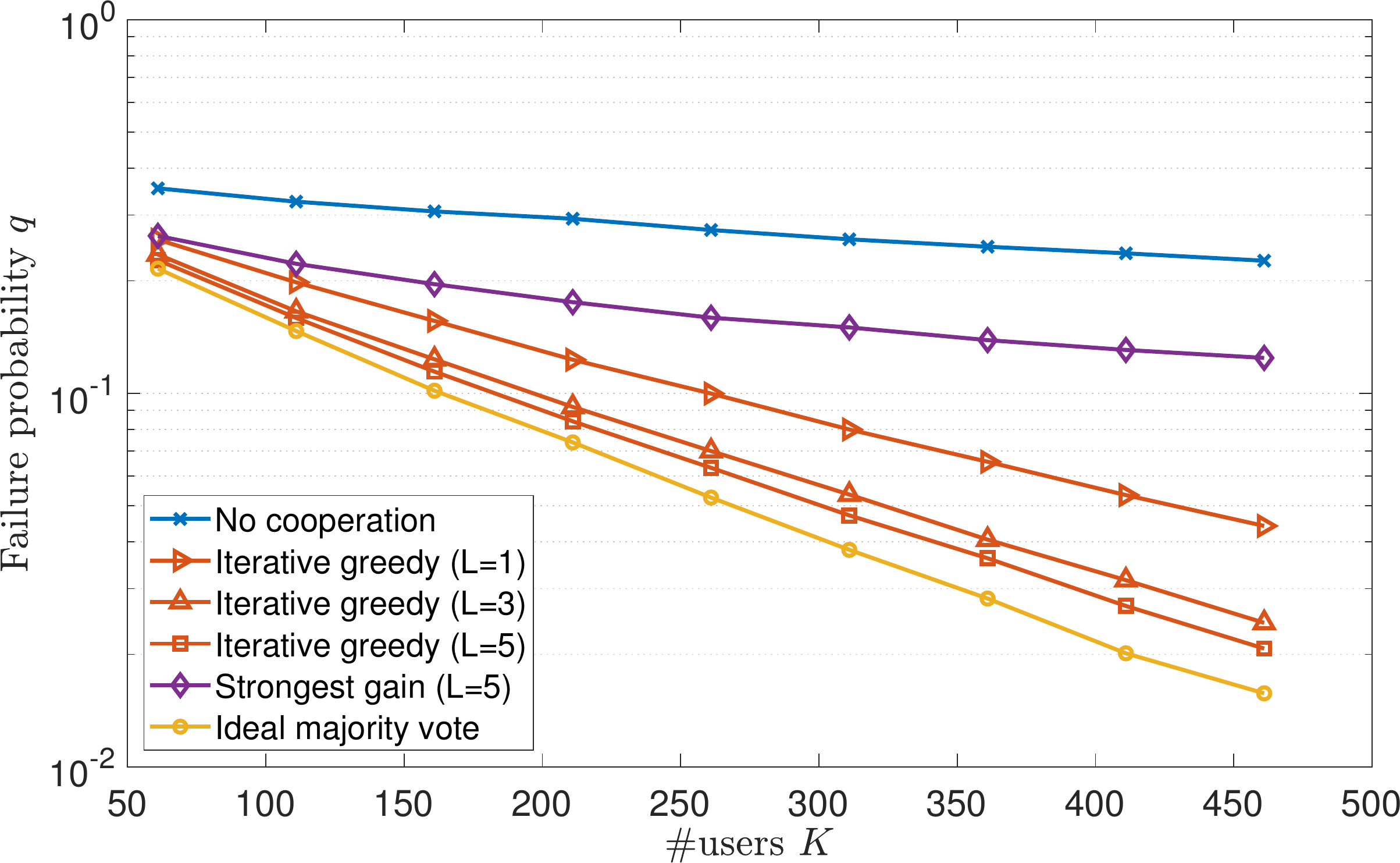}
    \caption{The failure probability versus number of users with a fixed
    size of cluster $K_C=9$. The other parameters are the same as in Fig.~\ref{fig: path loss}.}
    \label{fig:cooperation_error_curve}
\end{figure}

\section{Simulation Results}
We evaluate our relay selection
schemes in an edge learning system with
$K=54$ users. The cell radius is $R=1000$~m, path loss exponent $\alpha=3$, and
$r_0=10$ m in \eqref{eq:path_loss}. The number of subchannels is $M=1000$ and
4-QAM is adopted. The symbol power budget
is $P_s=-50$ dBW and the noise variance $N_0=80$ dBm.
We consider the MNIST digit recognition problem and use Algorithm \ref{alg:signsgd}
to train a multilayer perceptron (MLP) with one hidden layer of 64 units using ReLu activations. 
The training samples are distributed uniformly at random.
In our cooperative scheme, we assume there are $C=6$ clusters each containing
$K_C=9$ users.

As shown in Fig.~\ref{fig:accuracy},
the relay
selection scheme with the iterative greedy algorithm achieves
a very similar performance as the ideal majority vote. The difference of their final
test accuracies is smaller than 0.5\%. In contrast, the
system without cooperation converges more slowly and its
final test accuracy is near $1.2$\% less than the ideal case. This
shows that user clustering and relay selection can effectively combat the
channel fading impairments.

We also test the idea that the normalized SNR in \eqref{eq:normalized_SNR}
can be regarded as the participation rate of voting users.
By Monte Carlo simulations, we obtain that the normalized SNR is around
0.32 in the system without cooperation.
From our discussion in Section \ref{subsec:detection SNR}, the number of effective
users is $54\times 0.32\approx 17$. Therefore, we also plot
the test accuracy for ideal majority vote when $K=17$ (the dash line).
We can see that the performance is indeed close to
that of the system without cooperation, which is consistent with
our analysis.


\begin{figure}[!t]
    \centering
    \includegraphics[width=0.8\linewidth]{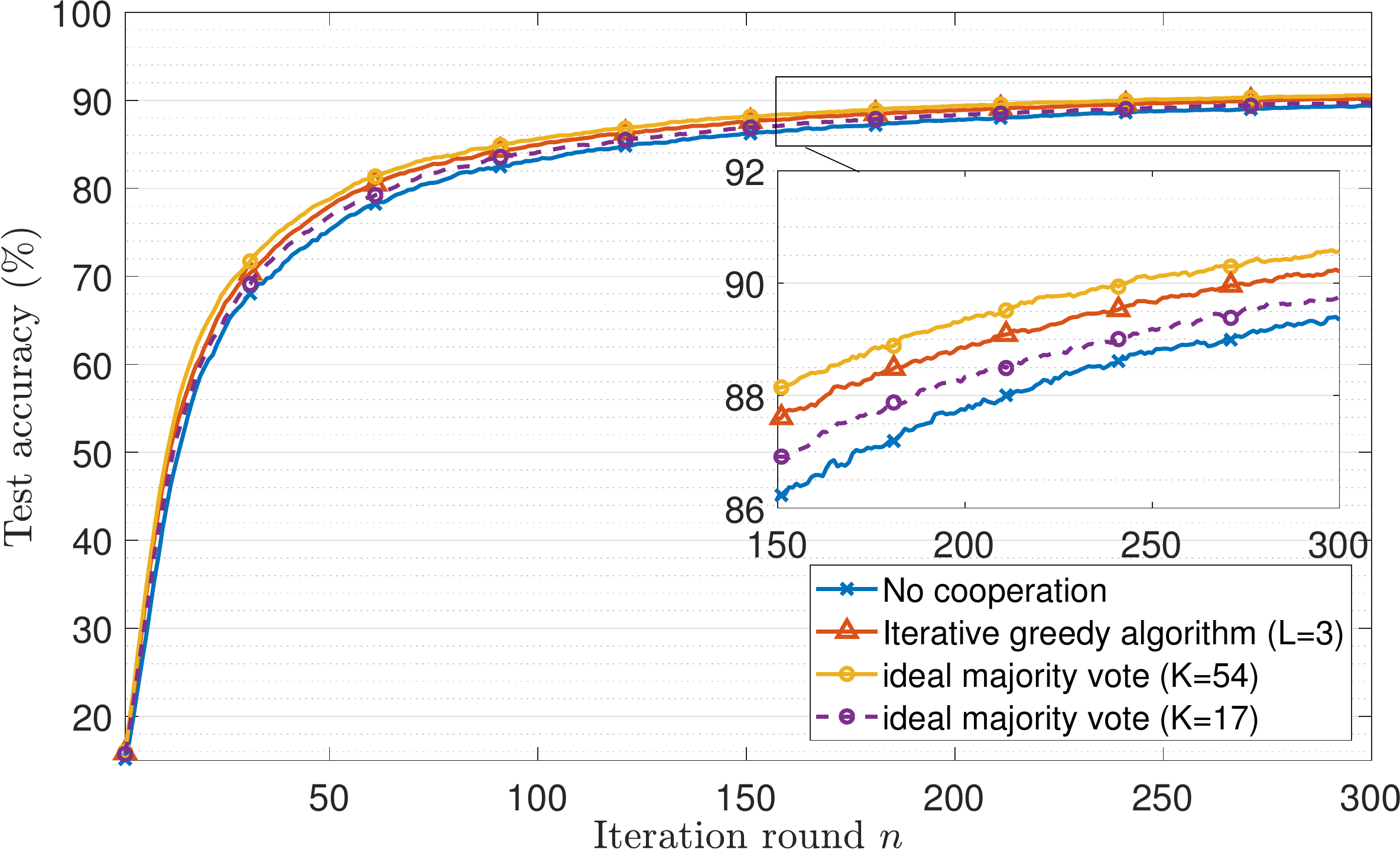}
    \caption{Test accuracy while training a MLP on MNIST dataset. The curves
    are the average results of 5 repeats.}
    \label{fig:accuracy}
\end{figure}

\section{Conclusions}
We proposed a new digital AirComp scheme for federated
edge learning systems that has less stringent requirements on the transmitters.
To characterize the failure probability of majority vote in our scheme, we derived an
upper bound and defined the normalized detection SNR. We found that the impact
of wireless
fading is equivalent to decreasing the number of voting users.
Furthermore, we designed a cluster-based cooperation scheme to combat wireless impairments
and proposed relay selection schemes based on the normalized detection SNR.
Simulations show that the relay selection scheme with the iterative greedy
algorithm can well exploit spatial diversity and achieve a comparable convergence
speed to the system with ideal channels.

\bibliographystyle{IEEEtran}
\bibliography{reference}

\end{document}